\documentclass[submission,copyright,creativecommons]{eptcs}
\providecommand{\event}{GraMSec 2014}
\usepackage{breakurl}
\usepackage[utf8x]{inputenc}
\usepackage{amsmath}
\usepackage{amssymb}
\usepackage{amsthm}
\usepackage{stmaryrd}
\usepackage{multicol}
\usepackage{algorithmicx}
\usepackage{algpseudocode}
\usepackage{algorithm}
\usepackage{graphicx}
\usepackage{wrapfig}
\usepackage{url}
\usepackage{tikz}
\usetikzlibrary{automata}
\usepackage{paralist}
\usepackage{comment}

\excludecomment{techreport}
\includecomment{paper}
\DeclareMathAlphabet{\mathcal}{OMS}{cmsy}{m}{n}

\newcommand{\view}{\mathcal{V}}
\newcommand{\activities}{\mathcal{A}}
\newcommand{\friends}{\activities_H^d}
\newcommand{\observers}{\activities_L^d}
\newcommand{\documents}{Docs}

\newtheorem{definition}{Definition}
\newtheorem{theorem}{Theorem}

\title{Possibilistic Information Flow Control for Workflow Management Systems\thanks{This research is supported by the Deutsche Forschungsgemeinschaft (DFG) under grant Hu737/5-1, which is part of the DFG priority programme 1496 ``Reliably Secure Software Systems.''}}
\author{Thomas Bauereiss \qquad\qquad Dieter Hutter
\institute{German Research Center for Artificial Intelligence (DFKI) \\
Bibliothekstr.~1\\
D-28359 Bremen, Germany}
\email{thomas.bauereiss@dfki.de \qquad\quad hutter@dfki.de}
}

\def\titlerunning{Possibilistic Information Flow Control for Workflow Management Systems}
\def\authorrunning{T. Bauereiss, D. Hutter}

\begin{document}

\maketitle 

\begin{abstract}
  In workflows and business processes, there are often security requirements on
  both the data, i.e.\ confidentiality and integrity, and the process, e.g.\
  separation of duty. Graphical notations exist for specifying both workflows
  and associated security requirements. We present an approach for formally
  verifying that a workflow satisfies such security requirements.  For
  this purpose, we define the semantics of a workflow as a state-event system
  and formalise security properties in a trace-based way, i.e.\ on an abstract
  level without depending on details of enforcement mechanisms such as
  Role-Based Access Control (RBAC).  This formal model then allows us to build
  upon well-known verification techniques for information flow control. We
  describe how a compositional verification methodology for possibilistic
  information flow can be adapted to verify that a specification of a
  distributed workflow management system satisfies security requirements on
  both data and processes.
\end{abstract}

\section{Introduction}
\label{sec:introduction}

Computer-supported workflows and business process automation are widespread in
enterprises and organisations.  Workflow management systems support the
enactment of such workflows by coordinating the work of human participants in
the workflow for human activities as well as by automatically executing
activities that can be mechanised.  Graphical notations such as BPMN allow for
the specification of workflows in an intuitive way.  In addition to the control
and data flows, there are typically various security requirements that need to
be considered during the design, implementation and execution of workflows. A
well-known security requirement on workflows is separation of duty for fraud
prevention \cite{clark_comparison_1987}. Confidentiality of data is another
important security requirement, e.g.\ the confidentiality of medical data from
non-medical personnel. These two can be seen as examples for different types of
security requirements. On the one hand, there are security requirements on
processes, i.e.\ constraints on the control flow and the authorisation of
users, and on the other hand, there are security requirements on data, i.e.\
constraints on the flow of information.  Several proposals to extend BPMN with
graphical notations for both kinds of security requirements exist
\cite{brucker.ea:securebpmn:2012,DBLP:journals/ieicet/RodriguezFP07,DBLP:journals/re/WolterM10}.

In this paper, we focus on the question of how the semantics of such a notation
can be defined and how to use them to formally verify both types of security
requirements.  We do this on an abstract level without having to refer to
details of enforcement mechanisms such as role-based access control (RBAC).
For this purpose, we model the behaviour of a workflow as a set of traces of
events, each representing a possible run of the workflow, and formalise our
security requirements in a declarative way as properties of such trace sets.
We map process requirements such as separation of duty to sets of allowed
traces, corresponding to safety properties \cite{alpern_recognizing_1987},
whereas we map requirements on data to information flow properties, which have
been extensively studied
\cite{sabelfeld_language-based_2003,mantel_possibilistic_2000,zakinthinos_general_1997,mclean_general_1996,focardi_classification_1995}.
This allows us to verify the absence not only of direct information leaks via
unauthorised access, but also of indirect information leaks via observing the
behaviour of the system. For example, if the control flow depends on a
confidential data item and an unauthorised user observes which path of the
control flow has been taken, they might be able to deduce the confidential
value of the data item.

The relation between possibilistic information flow and safety properties is
not trivial due to the refinement paradox, i.e.\ enforcing a safety property by
removing disallowed traces might introduce new information leaks
\cite{mantel_preserving_2001}. We discuss this relation for the case of
separation of duty, give sufficient conditions for the compatibility with
information flow properties, and show that these conditions are satisfied in
our example setting.

We build upon the MAKS framework for possibilistic information flow control
\cite{mantel_possibilistic_2000}, which is suitable for formulating and
verifying information flow policies at the specification level, and in which
many information flow properties from the literature can be expressed.  We
describe how a compositional verification methodology
\cite{hutter_security_2007} can be applied to verify our system models, which
has the advantage that we can split up the verification task into separate
verification tasks of the individual activities that make up the overall
workflow.

Essentially, our approach allows for the formal modelling of workflows and the
verification of security requirements on data and processes at a high level of
abstraction.  We have verified our results using the interactive theorem prover
Isabelle/HOL \cite{nipkow2002isabelle}. To improve practicality, future work
will focus on refinement approaches of these specifications towards concrete
implementations, while preserving as much as possible of the security
properties established on the abstract level.  The long-term goal of our work
is to facilitate the step-wise development of secure workflow management
systems, starting from an abstract specification, derived from a workflow
diagram, for example, then performing a series of refinement steps and
eventually arriving at a secure implementation.

The rest of this paper is structured as follows.  In the following subsection,
we present a running example of a workflow that we will use for illustration
throughout this paper. Section \ref{sec:model} introduces the system model. In
Section \ref{sec:policy}, we elaborate on modelling confidentiality and
separation-of-duty requirements, respectively. In Section
\ref{sec:verification}, we describe how existing techniques for compositional
verification of information flow properties can be applied and adapted for our
workflow systems. Section \ref{sec:related} discusses related work and Section
\ref{sec:conclusion} concludes the paper.

\subsection{Example Scenario}

As a running example, we use the workflow depicted in Figure
\ref{fig:workflow}, adapted from an industry use-case described in an
(unpublished) paper by A.~Brucker and I.~Hang.  It models a hiring process
including interviews and medical examinations.  The swimlanes represent two
departments of the organisation running the hiring workflow, the Human
Resources (HR) and the medical department. The placement of activities in the
swimlanes indicates the responsible department, and thus the authorised
employees. The input and output relations of activities are depicted as
directed flows of documents between activities.

\begin{figure}[htp]
  \begin{center}
    \includegraphics[width=\textwidth]{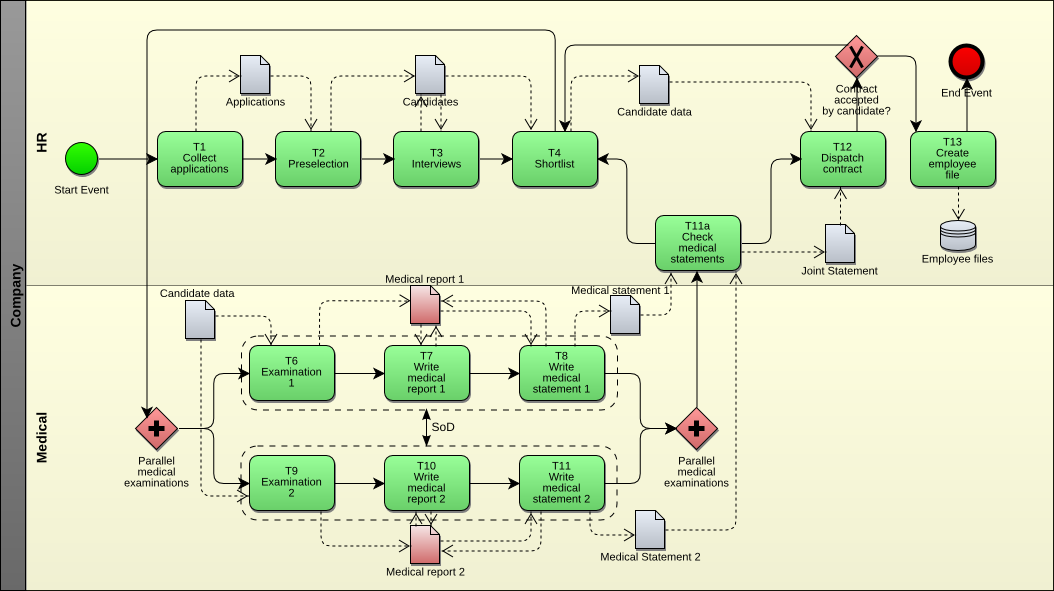}
  \end{center}
  \caption{Example workflow (adapted from A.~Brucker and I.~Hang, unpublished)}
  \label{fig:workflow}
\end{figure}

We use this workflow to illustrate security requirements with respect to both
data and process.  On the process side, we demand separation of duty between
the two medical examinations, i.e.\ they have to be performed by different
medical officers, such that no single medical officer can manipulate the hiring
process by rejecting unwanted candidates for fabricated medical reasons.
Regarding confidentiality requirements, we assume that the two medical reports
are highly confidential due to their potentially sensitive contents. In
particular, they are confidential for the employees in the HR department, who
should only be able to access information on a need-to-know basis, e.g.\ the
CVs of applicants.

In general, we assume there is a set of security domains, which are used to
classify documents exchanged between activities, and a flow policy that
specifies the allowed information flows between domains.  We assign a domain to
every document (a security classification) and to every employee (a security
clearance).  These classifications and clearances determine which users are
allowed to participate in which activities of the workflow.  For example,
employees with an $HR$ clearance are not allowed to participate in the
activities creating the medical report; otherwise, they would have direct
access to confidential medical data.  We formally define the constraints
regarding classifications, clearances and flow policies in Section
\ref{sec:confidentiality}.

Besides direct information flows via transfer of documents, we aim to control
indirect flows, where confidential information is deducible from observations
of the system behaviour. For example, an HR employee in the above workflow can
deduce whether the two medical officers agreed about the fitness of the
candidate by observing whether the workflow proceeds to activity 12 after the
medical examinations or reverts to activity 4. This is acceptable and actually
necessary in our scenario, as long as this is the only bit of information about
the medical condition of the candidate that can be deduced by HR personnel. Our
goal is to verify that the workflow indeed does not leak any additional medical
information to non-medical personnel. In the next section, we begin by formally
modelling the workflow, and then proceed to formalise the security
requirements.

\subsection{Preliminaries}

We briefly recall the definitions of (state-) event systems and security
predicates from the MAKS framework for possibilistic information flow
\cite{mantel_possibilistic_2000} that we use in this paper.  An event system
$ES = (E, I, O, Tr)$ is essentially a (prefix-closed) set of traces $Tr
\subseteq E^*$ that are finite sequences of events in the event set $E$.  The
disjoint sets $I \subseteq E$ and $O \subseteq E$ designate input and output
events, respectively.  We denote the empty trace as $\langle \rangle$, the
concatenation of traces $\alpha$ and $\beta$ as $\alpha .  \beta$, and the
projection of a trace $\alpha$ onto a set $E$ as $\alpha|_E$.  In the
composition $ES_1 \| ES_2$ of two event systems $ES_1$ and $ES_2$, input events
of one system matching output events of the other system are connected (and
vice versa) and thus become internal events of the composed system. The set of
traces is the set of interleaved traces of the two systems, synchronised on
events in $E_1 \cap E_2$: 
\[Tr(ES_1 \| ES_2) = \left\{ \alpha \in (E_1 \cup E_2)^* \mid \alpha|_{E_1} \in
Tr(ES_1) \land \alpha|_{E_2} \in Tr(ES_2) \right\}\]
A state-event system $SES = (E, I, O, S, s_0, T)$ has a set of states $S$, a
starting state $s_0 \in S$, and a transition relation $T \subseteq (S \times E
\times S)$.  The event system \emph{induced} by a state-event system has the
same sets of events and the set of traces that is enabled from the starting
state via the transition relation.

The MAKS framework defines a collection of basic security predicates (BSPs).
Many existing information flow properties from the literature can be expressed
as a combination of these BSPs. Each BSP is a predicate on a set of traces with
respect to a view $\view$. A view $\view = (V, N, C)$ on an event system $ES =
(E, I, O, Tr)$ is defined as a triple of event sets that form a disjoint
partition of $E$. The set $V$ defines the set of events that are visible for an
observer, $C$ are the confidential events, and the events in $N$ are neither
visible nor confidential. Notable examples for BSPs, that we will use in this
paper, are backwards-strict deletion ($BSD$) and backwards-strict insertion of
admissible confidential events ($BSIA$)\footnote{In
  \cite{mantel_composition_2002}, $BSIA$ is defined with an additional
  parameter $\rho$ that allows to strengthen the property by further specifying
  positions at which confidential events must be insertable. For simplicity, we
  choose to fix this parameter to $\rho_E$ in the notation of
  \cite{mantel_composition_2002}, i.e\ we only require confidential events to
  be insertable into a trace without interfering with observations if they are
  in principle admissible exactly at that point in the trace.},
defined in \cite{mantel_composition_2002} as follows:
\begin{align*}
  BSD_{\view}(Tr) \equiv & \, \forall \alpha, \beta \in E^*. \forall c \in C. \left( \beta . c . \alpha \in Tr \land \alpha |_C = \langle \rangle \right)  \\
  & \Rightarrow \exists \alpha' \in E^*. \left( \alpha' |_V = \alpha |_V \land \alpha' |_C = \langle \rangle \land \beta . \alpha' \in Tr \right) \\
  BSIA_{\view}(Tr) \equiv & \, \forall \alpha, \beta \in E^*. \forall c \in C. \left( \beta . \alpha \in Tr \land \alpha |_C = \langle \rangle \land \beta . c \in Tr \right)  \\
  & \Rightarrow \exists \alpha' \in E^*. \left( \alpha' |_V = \alpha |_V \land \alpha' |_C = \langle \rangle \land \beta . c . \alpha' \in Tr \right)
\end{align*}
Intuitively, the former requires that the occurrence of confidential events
must not be deducible, while the latter requires that the \emph{non}-occurrence
of confidential events must not be deducible. Technically, they are closure
properties of sets of traces. For example, if a trace in $Tr$ contains a
confidential event, then $BSD$ requires that a corresponding trace without the
confidential event exists in $Tr$ that yields the same observations. This means
the two traces must be equal with respect to visible $V$-events, while
$N$-events might be adapted to correct the deletion of the confidential event.

\section{System Model}
\label{sec:model}

In order to verify that a workflow satisfies given security requirements, we
need a formal model of workflows and their behaviour.  We first define our
notion of workflows. For simplicity, we omit aspects such as exceptions or
compensation handling, but our definition suffices for our purpose of
discussing the verification of security requirements for workflows.

\begin{definition}
  A workflow $W = (\activities, \documents, SF, MF, U)$ consists of
  \begin{compactitem}
    \item a set $\activities$ of activities, %\footnote{In BPMN terminology, ``activities'' are tasks or sub-processes.}
    \item a set $\documents$ of data items,
    \item a set $SF \subseteq (\activities \times \activities)$ of sequence flows, where $(a_1,
      a_2) \in SF$ represents the fact that upon completion of $a_1$, it may
      send a trigger to $a_2$ signalling it to start execution, and
    \item a set $MF \subseteq (\activities \times \documents \times \activities)$ of message flows, where
      $(a_1, d, a_2) \in MF$ represents data item $d$ being an output of
      activity $a_1$ and an input to $a_2$, and
    \item a set $U$ of users participating in the workflow.
  \end{compactitem}
\end{definition}
The sets $\activities$ and $\documents$ correspond to the nodes of a workflow
diagram such as Figure~\ref{fig:workflow}, while $SF$ and $MF$ correspond to
the solid and dashed edges, respectively.

We define the behaviour of workflows, not in a monolithic way, but in terms of
the behaviours of components representing activities communicating with each
other.  As we will show in Section~\ref{sec:verification}, this simplifies the
verification, because it allows us to use the decomposition methodology of
\cite{hutter_security_2007} to verify the security of the overall system by
verifying security properties of the subcomponents. We believe that such a
decomposition approach can help in scaling up verification of information flow
properties to larger systems.

Each activity $a$ is therefore modelled as a state-event system $SES_a = \left(
E_a, I_a, O_a, S_a, s_a^0, T_a \right)$ analogously to Definition $3$ of
\cite{hutter_security_2007}.  The set of events $E_a$ consists of events of the
form
\begin{compactitem}
  \item $Start_a(u)$, starting the activity $a$ and assigning it to the user $u
    \in U$,
  \item $End_a(u)$, marking the end of the activity,
  \item $Send_a(a', msg)$ and $Recv_a(a', msg)$, representing the sending (or
    receiving, respectively) of a message $msg$ from activity $a$ to activity
    $a'$ (or vice versa),
  \item $Setval_a(u, i, val)$ and $Outval_a(u, i, val)$, representing a user $u
    \in U$ reading (or setting, respectively) the value $val$ of data item $i$,
    and
  \item a set of internal events $\tau_a$.
\end{compactitem}
\begin{sloppypar}
We denote the set of events of a given activity $a \in \activities$ as $E_a$,
and the set of all events in a workflow as $E_W = \bigcup_{a \in \activities}
E_a$.  We denote the set of events of a given user $u \in U$ as $E_u = \left\{
Start_a(u) \mid a \in \activities \right\} \cup \left\{ Setval_a(u, i, val)
\mid a \in \activities, i \in \documents, val \in Val \right\} \cup \left\{
Outval_a(u, i, val) \mid a \in \activities, i \in \documents, val \in Val
\right\} \cup \left\{ End_a(u) \mid a \in \activities \right\}$, and the set of
all user interaction events as $E_U = \bigcup_{u \in U} E_u$.  The messages
between activities can have the form
\end{sloppypar}
\begin{compactitem}
  \item $Trigger$, used to trigger a sequence flow to a successor activity in
    the workflow,
  \item $Data(i, v)$, used to transfer the value $v$ for data item $i$, and
  \item $AckData(i)$, used to acknowledge the receipt of a data item.
\end{compactitem}
%The $AckData$ messages are due to the asynchronous communication platform.
Using separate messages for data and sequence flows is inspired by the BPMN
standard, which describes its (informal) execution semantics in terms of tokens
that are passed from one activity to the next, representing control flow
separately from data flows. In addition, this separation simplifies the
modelling of confidentiality, as it becomes straightforward to classify events
transporting $Data$ messages into confidential or non-confidential events based
on the classification of the data items they transport.

The local states of the activities include program variables such as a program
counter and a mapping $Mem \colon \documents \to Val$, storing the values of
data items.  After initialisation, the activity waits for messages from other
activities, transferring input data or triggering a sequence flow. When one (or
more) of the incoming transitions have been triggered, the activity internally
computes output messages (possibly via interaction with users), sends them via
the outgoing data associations, and triggers outgoing sequence flows.  In
Appendix~\ref{app:spec}, we formally specify two types of activities as
examples, namely user activities that allow users to read and write data items,
and gateway activities that make a decision on the control flow based on the
contents of their input data items.

Each of these state-event systems $SES_a$ induces a corresponding event system
$ES_a$. The overall system then emerges from the composition of these event
systems $ES_a$ for every activity $a \in \activities$, together with a
communication platform $ES_P$:
\[ES_{W} = \left( \|_{a \in \activities} ES_a \right) \| ES_P\]
We call $ES_W$ the \emph{workflow system} for the workflow $W$.  We reuse the
communication platform of \cite{hutter_security_2007}, which is formally
specified in Section 2.3 of \cite{hutter_security_2007}.  It asynchronously
forwards messages between the activities.  As we do not assume that it provides
guarantees regarding message delivery, its specification is very
simple.\footnote{However, it is possible to adapt the decomposition methodology
  to other communication models, e.g.\ providing some notion of reliability of
  message delivery or means for synchronous communication. For example, see
  Appendix~\ref{app:spec} for some remarks on guaranteeing a notion of ordered
message delivery.} Upon composition with the platform, the communication events
between the activities become internal events of the composed system. Only the
communication events with users remain input and output events. These events
form the user interface of the workflow system.

A simple version of our example workflow can be represented as a composition of
instances of the activity types specified in Appendix~\ref{app:spec}.  We can
represent the activity T11a in Figure~\ref{fig:workflow} as a gateway that
decides on the control flow based on the results of the medical examinations:
If they are positive, the workflow continues with dispatching the contract,
otherwise it goes back to selecting another candidate from the shortlist.  The
other activities essentially consist of users reading and generating documents,
so we can represent them as user activities.  Of course, these activities can
be enriched with further details, e.g.\ the interviews can be expanded to
subprocesses of their own, but we assume that this is handled in a subsequent
refinement step and consider only the abstract level in this paper.

\section{Security Policies}
\label{sec:policy}

\subsection{Confidentiality}
\label{sec:confidentiality}

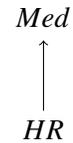
\begin{wrapfigure}{R}{0.35\textwidth}
  \centering
  \begin{tikzpicture}[shorten >=1pt,node distance=1.5cm,auto]
    \node (hr) {$HR$};
    \node (med) [above of=hr] {$Med$};
    \path[->] (hr) edge (med) ;
  \end{tikzpicture}
  \caption{Flow policy}
  \label{fig:confpol}
\end{wrapfigure}

We assign security domains from a set $\mathcal{D}$ of domains to the data
items exchanged between the activities of the workflow. We denote this domain
assignment function by $dom \colon \documents \rightarrow \mathcal{D}$. A flow
policy is a reflexive and transitive relation on domains and specifies from
which domains to which other domains information may
flow.\cite{myers_enforcing_2006} Note that, even though we focus on
confidentiality in this paper, also integrity requirements can be seen as a
dual to confidentiality and handled using information flow control. For
example, in \cite{myers_enforcing_2006} a lattice of combined security levels
is built as a product of a confidentiality lattice and an integrity lattice.
For our example workflow, we only require two confidentiality domains $HR$ and
$Med$.  The medical reports $MedReport1$ and $MedReport2$ created by activities
T6--T8 and T9--T11 in Figure~\ref{fig:workflow} are assigned to the $Med$
confidentiality domain, and other data items to the $HR$ domain. The example
flow policy states that information may flow from $HR$ to $Med$, but not vice
versa, i.e.\ $HR \leadsto Med$ and $Med \not \leadsto HR$ (see Figure
\ref{fig:confpol}).

Users read and write the contents of data items via the inputs and outputs of
activities they participate in. In order to exclude unwanted direct information
flows, we have to make sure that the classifications of the data items that
users work with are compatible with their clearances. A straightforward
approach is to enforce a Bell-LaPadula style mandatory access control. This can
be formulated in terms of classifications that are assigned to activities based
on the classifications of their inputs and outputs:

\begin{definition}
  \label{def:mac}
  An activity classification $cl_{\activities} \colon \activities \to
  \mathcal{D}$ is an assignment of domains to activities such that
  \begin{enumerate}
    \item for all input data items $i$ of an activity $a$, $dom(i) \leadsto
      cl_{\activities}(a)$, and
    \item for all output data items $i$ of an activity $a$ that may be assigned
      to untrusted users, $cl_{\activities}(a) \leadsto dom(i)$.
  \end{enumerate}
\end{definition}

We allow users to participate in an activity of a given classification only if
they have a matching clearance.  We denote the mapping of users to clearances
as $cl_U \colon U \to \mathcal{D}$.  The conditions in Definition~\ref{def:mac}
correspond to the Simple Security and $*$-Property of the Bell-LaPadula model,
respectively. Note that we relax the $*$-Property by allowing trusted users to
downgrade data items. Otherwise, we would not be able to assign a
classification to the activities T8 and T11 in our example workflow, because
they have high inputs (the medical reports) and low outputs (the statements
about the final result of examinations).  However, this specific flow of
information in the example is acceptable and necessary, because the output
should contain only the non-confidential final decision of the medical
officer\footnote{Which can be further enforced by allowing only Boolean values
as content of the low output.} required by the HR department, while the
detailed content of the medical reports remains classified as confidential.
Essentially, we admit inputs and outputs of trusted users to act as a channel
for declassification that is not formally controlled by our information flow
analysis.  It would be possible to model declassification more explicitly, e.g.
using intransitive flow policies \cite{mantel_information_2001}, but for
simplicity we choose this solution for this paper. The same approach is
followed in \cite{DBLP:journals/re/WolterM10}, for example. See
\cite{stork_downgrading_1975} for an early discussion of this approach to
downgrading and \cite{DBLP:journals/jcs/SabelfeldS09} for a general overview of
principles and dimensions of declassification.

Regardless of whether trusted users are present or not, we want to verify that
the system itself does not leak information about data items $i$ with
classification $dom(i) \not \leadsto d$ to users with clearance $d$. The set of
confidential events for a domain $d$ thus consists of events setting or reading
values of these data items, while events of activities whose classification is
allowed to flow into $d$ are considered to be potentially observable for users
in domain $d$:

\begin{definition}
  \label{def:view-d}
  Let $d \in \mathcal{D}$ be a domain.
  The security view on a workflow system $ES_W$ for $d$
  is defined as $\view_d = \left( V_d, N_d, C_d \right)$, where
  \begin{align*}
    V_d = & \bigcup_{cl_{\activities}(a) \leadsto d} E_a \\
    C_d = & \left\{ Setval_a(u, i, val) \mid \exists u \in U, i \in \documents, v \in Val. \; dom(i) \not \leadsto d \land cl_{\activities}(a) \not \leadsto d \right\} \\
       \cup & \left\{ Outval_a(u, i, val) \mid \exists u \in U, i \in \documents, v \in Val. \; dom(i) \not \leadsto d \land cl_{\activities}(a) \not \leadsto d \right\} \\
    N_d = & E \setminus (V_d \cup C_d)
  \end{align*}
\end{definition}

The set $C_d$ contains the confidential input and output events.\footnote{The
  set $C_d$ only contains events of activities $a$ with $cl_{\activities}(a)
  \not \leadsto d$, because activities with $cl_{\activities}(a) \leadsto d$
are considered to be visible, and the set of visible and confidential events
must be disjoint.} Note that we assume that confidential information enters the
system only via user input or output, and that the system does not generate
confidential information by itself (e.g.\ by generating cryptographic key
material). If that were the case, the corresponding system events would have to
be added to the set $C_d$.  Moreover, it is worth pointing out that we consider
certain other types of information to be \emph{non-confidential}. In
particular, the information whether an activity has been performed or not, or
the information which user has performed which activity is considered to be
non-confidential.  Again, such requirements could be captured by formulating
the security view accordingly.  For our setting, the above view reflects our
security requirement that the values of confidential data items should be kept
secret.  Hence, we use this view for the rest of this paper.

In Section~\ref{sec:verification}, we describe how to verify that a given
workflow system satisfies the security predicate $BSD_{\view_d} \land
BSIA_{\view_d}$ with respect to this view for every domain $d$. It expresses
that confidential user inputs and outputs can be deleted or inserted without
interfering with the observations of users in domain $d$.

\subsection{Separation of Duties}
\label{sec:sod}

As discussed in the introduction, separation of duties is another common
security requirement in workflow management systems. Separation of duties can
be formally defined as a safety property \cite{alpern_recognizing_1987}. The
``bad thing'' happens when the same user participates in two activities
constrained by separation of duty, hence we only allow traces where this does
not occur.

\begin{definition}
  \label{def:sod}
  Let $a, a' \in \activities$ be two activities.
  We call the set of traces
  \[\left\{ \alpha \in E_W^* \mid \forall u, u' \in U.\; \forall e_1, e_2 \in \alpha. (e_1 \in (E_a \cap E_u) \land e_2 \in (E_{a'} \cap E_{u'})) \rightarrow u \neq u' \right\}\]
  a separation-of-duty property $P_{SoD}^{a, a'}$.
\end{definition}

As we have modelled user assignment explicitly as events, this property can
also be characterised by requiring that
\begin{inparaenum}
  \item constrained activities are assigned to different users, and
  \item users may participate in an activity only after they have been assigned
    to it:
\end{inparaenum}
\begin{align*}
  P_{SoD}^{a, a'} \supseteq & \left\{ \alpha \in E_W^* \mid \forall u, u' \in U.\; Start_{a_1}(u) \in \alpha \land Start_{a_2}(u') \in \alpha \longrightarrow u \neq u' \right\} \\
  & \cap \left\{ \alpha \in E_W^* \mid \forall a \in \activities, u \in U, e \in (E_a \cap E_u).\; Start_a(u) \notin \alpha \longrightarrow e \notin \alpha \right\}
\end{align*}
A system with a set of traces $Tr$ and events $E$ satisfies such a property iff
$Tr \subseteq P_{SoD}^{a, a'}$.  In our example workflow, there are separation
of duty constraints between the activities belonging to the two medical
examinations ($T6$--$8$ in Figure~\ref{fig:workflow} one the one hand, and
$T9$--$11$ on the other hand). Hence, we want to enforce $P_{SoD}^{a, a'}$ for
the pairs $(a, a') \in \left\{ T6, T7, T8 \right\} \times \left\{ T9, T10, T11
\right\}$.

Similarly, other runtime-enforceable security policies
\cite{schneider_enforceable_2000} can be modelled as safety properties. In this
paper, we focus on the above notion of separation of duty as an example and
investigate its relation to information flow in Section~\ref{sec:safety}.

\section{Verification}
\label{sec:verification}

\subsection{Information Flow Security}
\label{sec:noninterference}

To ease the verification of the security of a workflow system, we decompose it
into the individual activities of the workflow and make use of the methodology
presented in \cite{hutter_security_2007} to verify the resulting distributed
system.  For each domain $d \in \mathcal{D}$, we verify that users in that
domain can learn nothing about information that is confidential for them.  The
first step of the methodology \cite{hutter_security_2007} is to partition the
activities into a set of low activities $\observers = \{ a \in \activities \mid
cl_{\activities}(a) \leadsto d \}$ that are (potentially) visible in domain
$d$ and a set of high activities $\friends = \{ a \in \activities \mid
cl_{\activities}(a) \not \leadsto d \}$ that are not visible and may handle
confidential information.\footnote{In \cite{hutter_security_2007}, the set
$\observers$ is called the set of observers, while $\friends$ is called the
set of friends. This might be a bit counterintuitive in our setting for
some readers, as the friends would be the activities that are \emph{not}
visible. To avoid confusion, we simply speak of low and high activities,
respectively.} It follows that $\view_d$ from Definition \ref{def:view-d} is a
global security view as defined in \cite[Definition 13]{hutter_security_2007},
i.e.\ the visible events are exactly the events of the low activities, the set
of confidential events is a subset of the events of the high activities, and
the remaining events are non-visible and non-confidential.

The second step is finding suitable local views $\view_d^a$ for high activities
$a \in \friends$ in order to verify that they do not leak confidential
information to low activities.  Hence, we cannot generally treat communication
events of these activities as $N$-events, as we did in the global view, but we
have to consider some of them as $V$-events (e.g.\ a high activity sending a
trigger or a declassified data item to a low activity) and some of them as
$C$-events (e.g.\ a high activity receiving a confidential data item).
Intuitively, this means we split each of these activities into a part that
visibly interacts with low activities and a part that handles confidential
data, and verify that the latter does not interfere with the former.
Technically, these local views satisfy certain constraints that allow us to
instantiate the compositionality result of \cite{hutter_security_2007}, as we
discuss below.
\begin{definition}
  Let $d \in \mathcal{D}$ be a domain, and $a \in \friends$ be a high activity
  for $d$.  Furthermore, let $\documents_d^C = \{i \in \documents \mid dom(i)
  \not \leadsto d\}$ denote the set of data items that are confidential for
  $d$. The \emph{local view for $a$} is defined as $\view_d^a = (V_d^a, N_d^a,
  C_d^a)$ with
  \begin{align*}
    V_d^a & = (I_a \cup O_a) \setminus \bigcup_{i \in \documents_d^C} E_i \\
    C_d^a & = \bigcup_{i \in \documents_d^C} \left( E_i \setminus \left\{ Send_a(b, m) \mid \exists v. \; m = Data(i, v) \lor m = AckData(i) \right\} \right) \\
    N_d^a & = E_a \setminus (V_d^a \cup C_d^a)
  \end{align*}
  where the set $E_i$ of high communication events containing data item $i$ is
  defined as
  \begin{align*}
  E_i = \big\{ e \mid & \exists b \in \friends, m, u, v. \; \left( m = Data(i, v) \lor m = AckData(i) \right) \\
  & \land \left( e = Send_a(b, m) \lor e = Recv_a(b, m) \lor e = Setval_a(u, i, v) \lor e = Outval_a(u, i, v) \right) \big\}
  \end{align*}
  Combining these local views, we define the \emph{composed view for $d$} as
  $\view_{d^+} = (V_{d^+}, N_{d^+}, C_{d^+})$ where

  \begin{tabular*}{0.95\textwidth}{@{\extracolsep{\fill} } lcr}
    $V_{d^+} = \bigcup_{a \in \friends} V_d^a \; \cup \; \bigcup_{a \in \observers} E_a$
    & $C_{d^+} = \bigcup_{a \in \friends} C_d^a$
    & $N_{d^+} = E_W \setminus (V_{d^+} \cup C_{d^+})$
  \end{tabular*}
\end{definition}

Note that the combined view $\view_{d^+}$ is \emph{stronger} than our global
view $\view_d$ in the sense that more events are considered confidential or
visible for an observer in domain $d$. Theorem 1 of
\cite{mantel_composition_2002} tells us that $BSD_{\view_{d^+}} \land
BSIA_{\view_{d^+}}$ for the stronger view implies $BSD_{\view_d} \land
BSIA_{\view_d}$.

Also note that all communication events with low activities are considered
visible, and that the forwarding of confidential data items from one high
activity to another is considered non-confidential. The justification for this
is that secrets enter and leave the subsystem of high activities through
communication with users and low activities, and the forwarding between high
activities can be considered as internal processing. Hence, we can use
communication events between high activities for correcting perturbations
caused by inserting or removing confidential user inputs.  We make use of this
fact in the proof of the following theorem, which states the security of
activities as we have specified them in Appendix~\ref{app:spec} in terms of the
transition relations $T_a^{gen}$, $T_a^{user}$ and $T_a^{gw(Cond)}$.

\begin{theorem}
  Let $W$ be a workflow, $d \in \mathcal{D}$ a domain and $SES_a$ for $a \in
  \activities$ an activity. If the transition relation of $SES_a$ is
  \begin{compactitem}
    \item $T_a^{gen} \cup T_a^{user}$, or
    \item $T_a^{gen} \cup T_a^{gw(Cond)}$ and $Cond$ does not depend on
      confidential data for $d$,
  \end{compactitem}
  then $BSD_{\view_d^a}(Tr_a) \land BSIA_{\view_d^a}(Tr_a)$ holds.
  \label{thm:local-security}
\end{theorem}

\begin{techreport}
  The proof of this and the following theorems can be found in
  Appendix~\ref{app:proofs}.
\end{techreport}
\begin{paper}
  The proof of this and the following theorems can be found in the extended
  version of this paper \cite{Bauereiss2013}.
\end{paper}
We use the unwinding technique \cite{DBLP:conf/esorics/Mantel00} for the proof.
Note that since the generic transition relation $T_a^{gen}$ and the
activity-specific transition relations are disjoint, we can partition this
proof into a generic part that covers the events and states used in
$T_a^{gen}$, and an activity-specific part.  Therefore, if we want to use a
different kind of activity than the ones specified in this paper, and we reuse
the generic part $T_a^{gen}$ of the transition relation, then we can also reuse
most of this proof.

The next step is to instantiate the compositionality result of
\cite{hutter_security_2007}, which states that the security of the overall
system with respect to the global security view is implied by the security of
the subsystems with respect to their local views. However, our local views do
not quite satisfy the requirement of being \emph{$C$-preserving} in the sense
of Definition 18 of \cite{hutter_security_2007}, because that definition
disallows $N$-events in the communication interface between subsystems.  Hence,
we slightly adapt the notion C-preserving views, allowing $Send$ events to be
in $N$:

\begin{definition}
  Let $\friends \subseteq \activities$ and $C \subseteq E_{\friends}$. A family
  $\left( \view^a \right)_{a \in \friends}$ of views $\view^a = (V_a, N_a,
  C_a)$ for $E_a$ is \emph{C-preserving for $C$} iff
  \begin{enumerate}
    \item $a \in \friends$ and $b \notin \friends$ implies $\forall m. Send_a(b, m) \in V_a  \land Recv_a(b, m) \in V_a$.
    \item $a, a' \in \friends$ implies
      \begin{enumerate}
	\item $Recv_{a'}(a, m) \in C_{a'}$ iff $Send_a(a', m) \not \in V_a$ and
	\item $Recv_{a'}(a, m) \in V_{a'}$ iff $Send_a(a', m) \in V_a$ % or $send_a(a', m) \notin C_a$.. ??
      \end{enumerate}
    %\item $E_P \cap E_a \subseteq C_a \cup V_a$ for all $a \in \Phi$.
    \item $C \cap E_a \subseteq C_a$ for all $a \in \Phi$.
    %\item $V \cap E_a = V_a$ for all $a \in \Phi$.
  \end{enumerate}
\end{definition}

As can be easily seen, our local views are C-preserving for the set of global
confidential events $C_d$ from Definition \ref{def:view-d}: communication with
low activities is visible, corresponding $Recv$ and $Send$ events are either
visible or non-visible (where non-visible $Recv$ events need to be
confidential, while the corresponding $Send$ events are allowed to be treated
as $N$-events), and events that are confidential in the global view are
confidential for the local views.

It turns out that the compositionality result of \cite{hutter_security_2007}
still holds for our weakened notion of C-preserving local views; a sufficient
(but not necessary) condition is that the subsystems satisfy not only $BSD$ (as
in \cite{hutter_security_2007}), but $BSD$ and $BSIA$, which our activities
happen to do.

\begin{theorem}
  Let $W$ be a workflow, $ES_W = \left( \|_{a \in \activities} ES_a \right) \|
  ES_P$ be a workflow system, $\view_d$ be a global security view for domain
  $d$, and $\left( \view_d^a \right)_{a \in \friends}$ be a family of local
  views that is C-preserving for $C_d$. If for all $a \in \friends$, $ES_a$
  satisfies $BSD_{\view_d^a} \land BSIA_{\view_d^a}$, then $ES_W$ satisfies
  $BSD_{\view_{d^+}} \land BSIA_{\view_{d^+}}$ and, therefore, $BSD_{\view_d}
  \land BSIA_{\view_d}$.
  \label{thm:comp-security}
\end{theorem}

Note that, if other kinds of activities than the ones from
Appendix~\ref{app:spec} should be part of the workflow, it is only required to
prove that their specifications also satisfy the security predicates for the
local views, in order to show that the overall workflow satisfies the
information flow security predicates.

We have formalised and verified our model and proofs using the interactive
theorem prover Isabelle \cite{nipkow2002isabelle}. Our development is based on
a formalisation of the MAKS framework developed by the group of Heiko Mantel at
TU Darmstadt (unpublished as of this writing). We intend to make our
formalisation publicly available when the MAKS formalisation is released.

Conceptually, the main difference between our workflow management systems and
the shopping mall system described in \cite{hutter_security_2007} lies in the
relation between users and the system. In the shopping scenario, there is a
one-to-one correspondence between users and software agents running in the
system.  Communication with the users happens only during initialisation, when
users write their preferences into the initial memory of their agents, which
run autonomously thereafter. In our workflow systems, the interaction is much
more dynamic, as multiple activities can be assigned to the same user at
runtime and there is ongoing communication between users and the system. This
has impact on the system model --- we introduced additional events for user
interaction --- and the construction of views. The partitioning into high and
low activities is based on classifications of data items and activities, and
access control has to ensure that only users with a matching clearance can
participate in an activity, so that our security views are actually in line
with the possible runtime observations of users. Despite these differences, we
have seen that the methodology of \cite{hutter_security_2007} can be applied
with small technical adjustments.

\subsection{Compatibility with Separation of Duties}
\label{sec:safety}

As described in Section \ref{sec:sod}, we can formalise constraints such as
separation of duty as safety properties. Having established information flow
security of our workflow system, we now ask whether these security properties
are preserved when enforcing separation of duty constraints.  In general, this
is not the case. Altering a system such that it satisfies a safety property can
be seen as a refinement, and it is well-known that possibilistic information
flow security is not preserved under refinement in general
\cite{mantel_preserving_2001}. Consider, for example, the security predicate
$BSIA$. Repeatedly inserting confidential events of different users into a
trace can exhaust the possible user assignments that would satisfy the
separation of duty constraints, thus deadlocking the process and making further
visible observations impossible.  We can, however, try to find sufficient
conditions under which information flow properties are preserved:

\begin{theorem}
  Let $ES = (E, I, O, Tr)$ be an event system and $\view = (V, N, C)$ be a view
  for $ES$. Let $E_a, E_a' \subseteq E$ be two disjoint sets of events
  corresponding to activities $a$ and $a'$, and let $P_{SoD}^{a, a'}$ be an SoD
  property.  Let $E_u \subseteq V \cup C$ be the communication events with a
  user $u$ and $E_U = \bigcup_{u \in U} E_u$ the set of all user events.  If
  \begin{enumerate}
    \item user assignment is non-confidential, i.e.\ there is a set $E^{assign}
      \subseteq E \setminus C$ of assignment events, and a user $u$ may only
      participate in an activity after having been assigned to it via an event
      from $E^{assign} \cap E_u$, or
    \item only confidential or only visible user I/O events of activities $a$
      and $a'$ are enabled in $ES$, i.e.\ there is a set $E^{disabled}
      \subseteq E$ of events that never occur in a trace of $ES$, and $V \cap
      (E_a \cup E_a') \cap E_U \subseteq E^{disabled}$ or $C \cap (E_a \cup
      E_a') \cap E_U \subseteq E^{disabled}$ holds, or
    \item the SoD constraint between $a$ and $a'$ is already enforced by $ES$,
      i.e.\ $Tr \subseteq P_{SoD}^{a, a'}$,
  \end{enumerate}
  then $BSD_{\view}(Tr) \land BSIA_{\view}(Tr)$ implies $BSD_{\view}(Tr \cap
  P_{SoD}^{a, a'}) \land BSIA_{\view}(Tr \cap P_{SoD}^{a, a'})$.
  \label{thm:sod-security}
\end{theorem}

In our running example, we can choose $E^{assign} = \left\{ Start_a(u) \mid u
\in U \right\}$ and apply the first case of the theorem for the workflow system
$ES_W$ and a view $\view_{d^+}$, because only the details of the results of the
medical examinations are confidential, not the information who carried out the
examinations.  Furthermore, in case $cl_{\activities}(a) \neq
cl_{\activities}(a')$, the mandatory access control described in
Section~\ref{sec:confidentiality} already enforces SoD statically, so the third
condition also applies.  In general, Theorem~\ref{thm:sod-security} gives us
sufficient conditions for the compatibility of SoD constraints and information
flow properties, taking into account the classifications of events that are
relevant for enforcing SoD.  Similar results could be developed for other
classes safety properties that are of interest in workflows, but we leave this
as future work.  Note that Theorem~\ref{thm:sod-security} is not specific to
workflow systems as specified in this paper. It can be applied to any system
where users perform different activities in the presence of separation of duty
constraints.

\section{Related Work}
\label{sec:related}

We build upon the MAKS framework for possibilistic information flow control
\cite{mantel_possibilistic_2000}, which is suitable for formulating and
verifying information flow policies at the specification level.  We have
focused on confidentiality of data from unauthorised employees within the
organisation, but in principle information flow control can be adapted to
different attacker models and security policies by choosing the security views
appropriately.  Furthermore, approaches have been proposed to take into account
factors such as communication over the Internet \cite{hutter_preserving_2007}
or encrypted communication channels \cite{hutter_possibilistic_2004}.  In
\cite{osborn_configuring_2000}, a connection between role-based access control
(RBAC) and mandatory access control is drawn, which might be adapted to enforce
the mandatory access control we described in Section~\ref{sec:confidentiality}
using RBAC mechanisms.

Early examples for workflow management systems with distributed architectures
include
\cite{DBLP:journals/dpd/AlonsoGKAAM96,muth_centralized_1998,schuster_client/server_1994}.
Later, computing paradigms with a similar spirit have emerged, e.g.
service-oriented architectures or cloud computing. We see these techniques and
standards as complementary to our work, as they can be used for the
implementation of our abstract specifications.

BPMN extensions to annotate business process diagrams with security annotations
can be found in
\cite{brucker.ea:securebpmn:2012,DBLP:journals/ieicet/RodriguezFP07,DBLP:journals/re/WolterM10}.
Closest to the security requirements considered by us comes the notation
proposed in \cite{DBLP:journals/re/WolterM10} that supports both the annotation
of activities with separation of duty constraints and the annotation of
documents and process lanes with confidentiality and integrity classifications
or clearances, respectively.

Several proposals for a formal semantics of workflow specifications can be
found in the literature. For example, \cite{wong_process_2008} maps BPMN
diagrams to CSP processes and describes how the formal semantics can be
leveraged to compare and analyse workflow diagrams, e.g.\ with respect to
consistency. It focuses on the control flow and does not model data flows.  In
\cite{yang_information_2010}, workflows are represented as statements in a
workflow description language, which is mapped to a representation as
hierarchical state machines. An information flow analysis algorithm is
described, but the actual information flow property that it checks is not
stated in a declarative, mechanism-independent way.
\cite{DBLP:conf/bpm/AccorsiL12} represents workflows as Petri nets and
describes an approach for information flow analysis. The focus is on keeping
the occurrence of tasks confidential, whereas our work focuses on the
confidentiality of the data that is processed in the workflow.  In
\cite{arsac_security_2011} and \cite{schaad_model-checking_2006}, workflows are
formalised as transition systems and model-checking is employed to verify
properties specified as LTL formulas. This is suitable to verify safety or
liveness properties, whereas the information flow predicates considered by us
can be seen as hyperproperties \cite{clarkson_hyperproperties_2010}.

\section{Conclusion}
\label{sec:conclusion}

Graphical notations such as BPMN are widely used for workflow specification. We
have presented an approach to formally model both the behaviour of a workflow
and the associated security requirements, and described how to apply the
decomposition methodology of \cite{hutter_security_2007} and how to verify a
distributed workflow management system with ongoing user interaction. We have
shown that, even though possibilistic information is in general not
refinement-closed, the enforcement of separation of duty is compatible with the
information flow security of the system under certain assumptions.

We have sketched how a simple version of our example workflow can be
represented as a composition of instantiations of the activity types specified
in Appendix~\ref{app:spec}.  As we have shown the security of these activities
in Theorem~\ref{thm:local-security}, we can use Theorem~\ref{thm:comp-security}
to derive the security of the composed system from the security properties of
the individual activities.  This demonstrates how instantiations of a type of
activities that has been proven secure once can be plugged into larger
workflows in a secure way.  Hence, we believe that this compositional approach
can help in making verification techniques for information flow scale to larger
workflow systems.  However, more work is needed before this approach can
actually be applied to realistic systems.  For example, tool support for
translating a more realistic subset of BPMN to our system model would be a
major step in this direction, which would also help us evaluate our approach
with a sample of existing workflows.

Moving from an abstract specification towards the implementation level is
another important direction of future work.  This paper deals with workflows on
a high level of abstraction.  We intend to work on notions of
security-preserving refinement that allow us to expand abstract activities in a
workflow into more concrete subprocesses and refine the behaviour of atomic
activities towards an executable implementation.  There is a large body of
existing work that we can build upon for this purpose, such as action
refinement for replacing atomic events on the abstract level with sequences of
more concrete events \cite{hutter_possibilistic_2006}, switching between
event-based and language-based notions of information flow
\cite{DBLP:journals/jcs/MantelS03}, or directly generating executable code from
specifications \cite{haftmann_code_2007}.  In the long term, we hope that these
decomposition and refinement techniques will contribute to making the step-wise
development of secure workflow systems from workflow diagrams to executable
code more scalable and efficient.

\paragraph{Acknowledgements} We thank Richard Gay, Sylvia Grewe, Steffen Lortz,
Heiko Mantel and Henning Sudbrock for providing a formalisation of the MAKS
framework in Isabelle/HOL that allowed us to verify our main results in
Isabelle, and the anonymous reviewers for helpful comments on the paper.

\bibliographystyle{eptcs}
\bibliography{references}

\appendix

\section{Specification of Activities}
\label{app:spec}

In this appendix, we give a formal specification of the behaviour of our
activities using PP-statements. In this formalism, the transition relation of a
state-event system is specified by listing pre- and post-conditions on the
state for each event (see Section 2.1 of \cite{hutter_security_2007} for a
formal semantics).

\algnewcommand\PP[2]{\item[#1;] \mbox{\textbf{affects:} #2}}
\algnewcommand\Pre{\item[\textbf{Pre: }]}
\algnewcommand\Post{\item[\textbf{Post: }]}

\begin{figure}[htbp]
  \hrule
  \begin{multicols}{2}
    \parbox{0.45\textwidth}{
      \begin{algorithmic}
	\PP{$Recv_a(b, Data(i, v))$}{$Mem$, $AQueue$}
	\Pre $pc = 0$, $(b, i, a) \in MF$, $(b, i) \notin AQueue$
	\Post $Mem'(i) = v$, $AQueue' = AQueue \cup \{(b, i)\}$
	\Statex
      \end{algorithmic}
    }

    \parbox{0.45\textwidth}{
      \begin{algorithmic}
	\PP{$Send_a(b, AckData(i))$}{$AQueue$}
	\Pre $pc = 0$, $(b, i) \in AQueue$
	\Post $AQueue' = AQueue \setminus (b, i)$
	\Statex
      \end{algorithmic}
    }

    \parbox{0.45\textwidth}{
      \begin{algorithmic}
	\PP{$Recv_a(b, Trigger)$}{$TriggeredBy$}
	\Pre $pc = 0$
	\Post $TriggeredBy' = b$
	\Statex
      \end{algorithmic}
    }

    \parbox{0.45\textwidth}{
      \begin{algorithmic}
	\PP{$\tau_a^{Active}$}{$pc$}
	\Pre $pc = 0$, $TriggeredBy \neq \bot$, $AQueue = \emptyset$
	\Post $pc' = 1$
	\Statex
      \end{algorithmic}
    }

    \parbox{0.45\textwidth}{
      \begin{algorithmic}
	\PP{$\tau_a^{SendData}$}{$pc$, $MQueue$}
	\Pre $pc = 2$
	\Post $pc' = 3$, $MQueue' = \left\{ (b, i) \mid (a, i, b) \in MF \land Mem(i)
	  \neq \bot \right\}$
	\Statex
      \end{algorithmic}
    }

    \parbox{0.45\textwidth}{
      \begin{algorithmic}
	\PP{$Send_a(b, Data(i, v))$}{$MQueue, AQueue$}
	\Pre $pc = 3$, $(b, i) \in MQueue$
	\Post $MQueue' = MQueue \setminus \{(b, i)\}, AQueue' = AQueue \cup \{(b, i)\}$
	\Statex
      \end{algorithmic}
    }

    \parbox{0.45\textwidth}{
      \begin{algorithmic}
	\PP{$Recv_a(b, AckData(i))$}{$AQueue$}
	\Pre $pc = 3$, $(b, i) \in AQueue$
	\Post $AQueue' = AQueue \setminus \{(b, i)\}$
	\Statex
      \end{algorithmic}
    }

    \parbox{0.45\textwidth}{
      \begin{algorithmic}
	\PP{$\tau_a^{AckTimeout}$}{$AQueue$}
	\Pre $pc = 3$ %, $MQueue = \emptyset$
	\Post $AQueue' = \emptyset$
	\Statex
      \end{algorithmic}
    }

    \parbox{0.45\textwidth}{
      \begin{algorithmic}
	\PP{$\tau_a^{SendTriggers}$}{$pc$, $SQueue$}
	\Pre $pc = 3$, $MQueue = \emptyset$, $AQueue = \emptyset$
	\Post $pc' = 4$, $SQueue' = \left\{ b \mid (a, b) \in SF \right\}$
	\Statex
      \end{algorithmic}
    }

    \parbox{0.45\textwidth}{
      \begin{algorithmic}
	\PP{$Send_a(b, Trigger)$}{$SQueue$}
	\Pre $pc = 4$, $b \in SQueue$
	\Post $SQueue' = SQueue \setminus \{b\}$
      \end{algorithmic}
    }
  \end{multicols}
  \hrule
  \caption{PP-statements of generic transition relation $T_a^{gen}$}
  \label{alg:generic-activity}
\end{figure}

We specify the behaviour of our activities in two parts. The PP-statements in
Figure~\ref{alg:generic-activity} specify the \emph{generic} part of the
behaviour of activities, i.e.\ the communication with other activities in order
to exchange data items and trigger sequence flows. For this purpose, it
maintains program variables $MQueue$ (which data items still have to be sent),
$AQueue$ (which data items still have to be acknowledged), $SQueue$ (which
triggers still have to be sent), $TriggeredBy$ (whether and from where a
trigger has been received), and $User$ (to which user this activity is
assigned). The program counters $0$, $3$ and $4$ correspond to the phases of
waiting for inputs and triggers, sending outputs, and sending triggers,
respectively.

When the program counter reaches $1$, an \emph{activity-specific} transition
relation takes over in order to perform the actual activity.  In our simple
example workflow, we only need two kinds of activities, namely user
input/output and gateways (deciding on the control flow based on a condition
$Cond$ on input data). The latter continues the workflow with that activity $b$
for which $Cond(b, Mem)$ evaluates to true. These two kinds of activities are
specified in Figures \ref{alg:user-activity} and \ref{alg:gateway},
respectively.  We denote the transition relations induced by the PP-statements
in Figures \ref{alg:generic-activity}, \ref{alg:user-activity}, and
\ref{alg:gateway} as $T_a^{gen}$, $T_a^{user}$, and $T_a^{gw(Cond)}$,
respectively.  The overall transition relation of an activity is the union of
$T_a^{gen}$ and an activity-specific transition relation.

\begin{figure}[htbp]
  \hrule
  \begin{multicols}{2}
    \parbox{0.45\textwidth}{
      \begin{algorithmic}
	\PP{$Start_a(u)$}{$User$}
	\Pre $pc = 1$, $User = \bot$, $cl_U(u) = cl_{\activities}(a)$
	\Post $User' = u$
	\Statex
      \end{algorithmic}
    }

    \parbox{0.45\textwidth}{
      \begin{algorithmic}
	\PP{$Setval_a(u, i, v)$}{$Mem$}
	\Pre $pc = 1$, $User = u$
	\Post $Mem'(i) = v$
	\Statex
      \end{algorithmic}
    }

    \parbox{0.45\textwidth}{
      \begin{algorithmic}
	\PP{$Outval_a(u, i, v)$}{}
	\Pre $pc = 1$, $User = u$, $Mem(i) = v$
	\Statex
      \end{algorithmic}
    }

    \parbox{0.45\textwidth}{
      \begin{algorithmic}
	\PP{$End_a(u)$}{$pc$}
	\Pre $pc = 1, User = u$
	\Post $pc' = 2$
      \end{algorithmic}
    }
  \end{multicols}
  \hrule
  \caption{PP-statements of transition relation $T_a^{user}$ for user activities}
  \label{alg:user-activity}
\end{figure}

\begin{figure}[htbp]
  \hrule
  \begin{algorithmic}
    \PP{$Send_a(b, Trigger)$}{$pc$}
    \Pre $pc = 1$, $Cond(b, Mem) = \top$, $(a, b) \in SF$
    \Post $pc' = 5$
  \end{algorithmic}
  \hrule
  \caption{PP-statement of transition relation $T_a^{gw(Cond)}$ for gateways}
  \label{alg:gateway}
\end{figure}

After completion of the activity has been signalled by setting the program
counter to $2$, the generic transition relation takes control again and starts
sending output data items to the designated receivers. It makes sure that they
have been received by waiting for acknowledgements, and afterwards proceeds by
sending triggers to the successor activities in the workflow. An exception to
this rule is if a receiver fails to send an acknowledgement; in this case the
$\tau_a^{AckTimeout}$ event can be used to signal a timeout and proceed with
the workflow. This is important for security, because otherwise a confidential
activity could block the progress of the workflow by refusing to acknowledge a
data item.

Of course, other modelling decisions are possible to solve this problem. As an
alternative, we have also modelled and verified a system specification where
the communication platform guarantees causal delivery of messages, i.e.
messages from one activity to another are always received in the order that
they are sent. This would make acknowledgements unnecessary, because an
activity could always be sure that a trigger message is received after all data
items, if the messages are sent in this order. However, this shifts complexity
from the individual activities to the communication platform and the interface,
and it turns out that this makes the proof of compositionality more laborious.
Essentially, we had to prove an additional security predicate $FCIA$ for the
platform and the activities together with several additional side conditions on
the local views in order to obtain compositionality. In this paper, we
therefore present the above model with explicit acknowledgements for
simplicity. However, we intend to further investigate the implications of
different guarantees provided by the communication platform in future work.

\begin{techreport}

\section{Proofs}
\label{app:proofs}

\begin{proof}[Proof of Theorem~\ref{thm:local-security}]
  Intuitively, we can convince ourselves that every possible confidential event
  can be removed from a trace of one of our activities or inserted at
  admissible locations without interfering with visible behaviour. For $Outval$
  events, this follows from the fact that they do not modify the state.
  Confidential $Setval$ and $Recv$ events modify the memory content only of
  confidential data items, and might make it necessary to send a corresponding
  output or acknowledgement afterwards. However, these $Send$ events are
  allowed to be inserted into or removed from the trace by our security
  predicates, because we classified $Send$ events of confidential data items as
  $N$ events.

  Formally, we prove the security of our activities using the unwinding
  technique of \cite{DBLP:conf/esorics/Mantel00}.  We define the following
  unwinding relation such that two states are related if they allow the same
  visible behaviour.
  \begin{align*}
    s \approx_d s' \Longleftrightarrow& \; pc(s) = pc(s') \\
    & \land \left( \forall i \in \documents_d. \; Mem(s)(i) = Mem(s')(i) \right) \\
    & \land \left( \forall b \in \activities, i \in \documents. \; \left(b \in \observers \lor i \in \documents_d \right) \longrightarrow \left( (b, i) \in MQueue(s) \longleftrightarrow (b, i) \in MQueue(s') \right) \right) \\
    & \land \left( \forall b \in \activities, i \in \documents. \; \left(b \in \observers \lor i \in \documents_d \right) \longrightarrow \left( (b, i) \in AQueue(s) \longleftrightarrow (b, i) \in AQueue(s') \right) \right) \\
    & \land SQueue(s) = SQueue(s') \\
    & \land TriggeredBy(s) = TriggeredBy(s') \\
    & \land User(s) = User(s')
  \end{align*}
  where $\documents_d = \left\{ i \in \documents \mid dom(d) \leadsto d
  \right\}$ denotes the set of documents that are visible in domain $d$.

  Theorem 3 of \cite{DBLP:conf/esorics/Mantel00} tells us that we can prove
  $BSD$ and $BSIA$ for our local view by proving the following unwinding
  conditions:
  \begin{align*}
    osc \equiv & \forall s_1, s_1' \in S. \, s_1 \approx_d s_1' \longrightarrow \forall e \in (E \setminus C_a^d). \forall s_2 \in S. \Big[ (s_1, e, s_2) \in T_a \longrightarrow \\
    & \exists \gamma \in (E \setminus C_a^d)^*. \exists s_2' \in S. \left( \gamma|_{V_a^d} = \langle e \rangle |_{V_a^d} \land s_2' \in succ(s_1', \gamma) \land s_2 \approx_d s_2' \right) \Big] \\ \\
    lrf \equiv & \forall s, s' \in S. \forall c \in C_d^a. \left( \left( reachable(s) \land (s, c, s') \in T_a \right) \longrightarrow s' \approx_d s \right) \\ \\
    lrb \equiv & \forall s \in S. \forall c \in C_d^a. \left( \left( reachable(s) \land s \in pre(c) \right) \longrightarrow \exists s' \in S. \; \left( (s, c, s') \in T_a \land s \approx_d s' \right) \right)
  \end{align*}
  where $succ(s, \gamma)$ denotes the set of states that the system can reach
  from state $s$ via the sequence of events $\gamma$, $reachable(s)$ is true
  iff $s$ is reachable from the starting state via some trace, and $pre(c)$ is
  the set of states that satisfy the preconditions of event $c$.

  We first prove $osc$. We choose arbitrary but fixed states $s_1$, $s_2$,
  $s_1'$ and a nonconfidential event $e$. We have to prove that, if $s_1
  \approx_d s_1'$ holds, then the nonconfidential transition $(s_1, e, s_2)$
  can be simulated in the state $s_1'$ by a sequence of events $\gamma$ that
  yields the same visible observations. We proceed by performing a case
  distinction on $e$ and finding a suitable witness $\gamma$ for each case:
  \begin{itemize}
    \item In case $e$ is of the form $Recv_a(b, Trigger)$, $\tau_a^{SendData}$,
      $\tau_a^{AckTimeout}$, $Send_a(b, Trigger)$, $Start_a(u)$, $Setval_a(u,
      i, v)$, $Outval_a(u, i, v)$, or $End_a(u)$, then the assumption $s_1
      \approx_d s_1'$ implies that the preconditions of $e$ are satisfied not
      only in $s_1$, but also in $s_1'$.  Furthermore, the postconditions of
      these events imply that the relation $s_2 \approx_d s_2'$ holds for the
      two successor states $s_2$ and $s_2' \in succ(s_1', \langle e \rangle)$.
      Hence, we can choose $\gamma = \langle e \rangle$ as the witness for
      $osc$.
    \item The same holds for the case where $e$ is of the form $Recv_a(b,
      Data(i, v))$ or $Recv_a(b, AckData(i))$, because it follows from $e
      \notin C_a^d$ and our definition of local views that $b$ must be an
      observer or $i$ a nonconfidential data item, i.e.~$b \in \observers \lor
      i \in \documents_d$. With $s_1 \approx_d s_1'$ we then have $(b, i) \in
      AQueue(s_1)$ iff $(b, i) \in AQueue(s_1')$, hence the preconditions of
      $e$ are also satisfied in $s_1'$ and we can again choose $\gamma =
      \langle e \rangle$ as witness.
    \item If $e = Send_a(b, AckData(i))$, we know that $(b, i) \in
      AckData(s_1)$ from $(s_1, e, s_1') \in T_a$.
      \begin{itemize}
	\item If $(b, i) \in AckData(s_1')$ holds, then $e$ is also enabled in
	  $s_1'$ and we again have $\gamma = \langle e \rangle$ as witness.
	\item If $(b, i) \notin AckData(s_1')$, then $s_1 \approx_d s_1'$
	  implies $b \in \friends \land i \notin \documents_d$. With our
	  definition of local views, this implies $e \in N_a^d$, hence $\langle
	  e \rangle |_{V_a^d} = \langle \rangle$. Furthermore, the
	  postcondition of $e$ then implies that $s_2 \approx_d s_1$ and, by
	  transitivity, $s_2 \approx_d s_1'$. Ultimately, we have $\gamma =
	  \langle \rangle$ as a suitable witness.
      \end{itemize}
    \item Let $e = Send_a(b, Data(i, v))$.
      \begin{itemize}
	\item If $i \in \documents_d$, then $Mem(s_1)(i) = Mem(s_1')(i)$ and
	  $(b, i) \in MQueue(s_1')$ due to $s_1 \approx_d s_1'$. Hence, we can
	  use $\gamma = \langle e \rangle$ as witness.
	\item If $i \notin \documents_d$, then $b$ must be a high activity,
	  because $dom(i) \leadsto cl_{\activities}(b)$ due to
	  Definition~\ref{def:mac}, and with $dom(i) \not \leadsto d$ it
	  follows that $cl_{\activities}(b) \not \leadsto d$ must hold due to
	  transitivity of $\leadsto$. Hence, $e \in N_a^d$, and analogously to
	  the second sub-case of $e = Send_a(b, AckData(i))$ above, we obtain
	  $\gamma = \langle \rangle$ as a suitable witness.
      \end{itemize}
    \item If $e = \tau_a^{Active}$, then $e$ may not be enabled in $s_1'$
      because there are still some acknowledgements waiting in $AQueue(s_1')$.
      Due to $s_1 \approx_d s_1'$ and $AQueue(s_1) = \emptyset$, these
      acknowledgements must be for confidential data items and directed to high
      activities, i.e.\ $\forall (b, i) \in AQueue(s_1'). \, b \in \friends
      \land i \notin \documents_d$. The corresponding events $Send_a(b,
      AckData(i))$ must then be in $N_a^d$, according to our definition of
      local views. Hence, there is a sequence $\delta$ of these events and a
      successor state $s_1'' \in succ(s_1', \delta)$ such that $\delta
      |_{V_a^d} = \langle \rangle$ and $s_1''$ is equal to $s_1'$ except
      $AQueue(s_1'') = \emptyset$. Then, $e = Start_a(u)$ is enabled in
      $s_1''$, and we obtain $\gamma = \delta .  \langle e \rangle$ as witness.
    \item Similarly, if $e = \tau_a^{SendTriggers}$, then $e$ may not be
      enabled in $s_1'$ because $MQueue(s_1')$ or $AQueue(s_1')$ still contain
      confidential data items to be sent or acknowledgements waiting to be
      received. Analogously as for $Start_a(u)$ above, we can obtain a sequence
      $\delta$ of events $Send_a(b, Data(i, v))$ with $(b, i) \in MQueue(s_1')$
      that sends the remaining data items, such that $\delta |_{V_d^a} =
      \langle \rangle$ and $MQueue(s_1'') = \emptyset$ for $s_1'' \in
      succ(s_1', \delta)$. The remaining acknowledgements can be cancelled by
      inserting an event $\tau_a^{AckTimeout}$. Combined, we obtain a witness
      $\gamma = \delta . \langle \tau_a^{AckTimeout} \rangle . \langle e
      \rangle$.
  \end{itemize}

  In order to prove $lrf$, we choose arbitrary but fixed states $s$, $s'$ and a
  confidential event $c \in C_d^a$. Due to our definition of local views, the
  confidential event $c$ must have one of the forms $Recv_a(b, Data(i, v))$,
  $Recv_a(b, AckData(i))$, $Setval_a(u, i, v)$ or $Outval_a(u, i, v)$ with $b
  \in \friends \land i \notin \documents_d$. In each of these cases, the
  contents of the memory and the queues for nonconfidential data items and low
  activities are unaffected. Hence, $s' \approx_d s$ holds.

  Finally, $lrb$ is implied by $lrf$ together with the facts that our unwinding
  relation is symmetric and $s \in pre(c) \longleftrightarrow \exists s'. \,
  (s, c, s') \in T_a$ due to our specification of the transition relation via
  PP-statements.
\end{proof}

\begin{proof}[Proof of Theorem~\ref{thm:comp-security}]
  The proof proceeds as in \cite{hutter_security_2007}, with the difference
  that we allow $Send$ events to be $N$-events in local views, i.e.\ they can
  be used for corrections of perturbations in confidential inputs.  The
  composition of high activities (friends) with the rest of the system is still
  well-behaved, because the second case of Definition 6.3.6 of
  \cite{DBLP:phd/de/Mantel2004} is satisfied: The subsystem obtained by
  composing a set of agents $X \subseteq \mathcal{A}$ and the platform is total
  in $Send$-events of agents $a \notin X$ and satisfies $BSIA$.
\end{proof}

\begin{proof}[Proof of Theorem~\ref{thm:sod-security}]
  We first prove that $BSD$ is preserved when restricting the set of traces
  from $Tr$ to $Tr \cap P_{SoD}^{a, a'}$.  $BSD$ can only be violated if there
  is a confidential event $c \in C$, a trace $\beta . c . \alpha \in Tr$ with
  $\alpha |_{C} = \langle \rangle$, and a perturbed and corrected trace $\beta
  . \alpha' \in Tr$ with $\alpha' |_{V} = \alpha |_{V}$ and $\alpha' |_{C} =
  \langle \rangle$, such that $\beta .  c . \alpha \in P_{SoD}^{a, a'}$, but
  $\beta .  \alpha' \notin P_{SoD}^{a, a'}$. With Definition~\ref{def:sod},
  this means there must be two distinct users $u$ and $u'$ and events $e_1 \in
  (E_a \cap E_u)$ and $e_2 \in (E_{a'} \cap E_{u'})$ such that $e_1$ and $e_2$
  both occur in $\beta . \alpha'$, but not both in $\beta . c . \alpha$.
  However, as by assumption all user interaction events are either visible or
  confidential in $\view$, i.e.\ $E_U \subseteq V \cup C$, and $\alpha |_{V
  \cup C} = \alpha' |_{V \cup C}$, this leads to a contradiction.  Hence, $BSD$
  must hold for the refined system.

  For $BSIA$ to be violated, there must be a trace $\beta . \alpha \in Tr$ with
  $\alpha |_{C} = \langle \rangle$, a confidential event $c \in C$ with $\beta
  . c \in Tr$, and a perturbed and corrected trace $\beta . c . \alpha' \in Tr$
  with $\alpha' |_{V} = \alpha |_{V}$ and $\alpha' |_{C} = \langle \rangle$,
  such that $\beta .  \alpha \in P_{SoD}^{a, a'}$ and $\beta . c \in
  P_{SoD}^{a, a'}$,  but $\beta . c .  \alpha' \notin P_{SoD}^{a, a'}$. Again,
  note that the subsequences of user interaction events in $\alpha$ and
  $\alpha'$ are equal due to $E_U \subseteq V \cup C$, i.e.\ $\alpha |_{E_U} =
  \alpha' |_{E_U}$. We perform a case distinction on the three conditions of
  Theorem~\ref{thm:sod-security}:
  \begin{enumerate}
    \item Assume that user assignments are nonconfidential. Hence, $c \in C$ is
      not a user assignment event. If $c$ is a user interaction event, i.e.\
      $\exists a, u. \, c \in E_u \cap E_a$, then from $\beta . c \in
      P_{SoD}^{a, a'}$ we have that an assignment event $e \in E^{assign} \cap
      E_a \cap E_u$ must have occurred in $\beta$, i.e.\ $a$ has already been
      assigned to $u$ before $c$ has occurred. If $c$ is not a user interaction
      event, it is irrelevant for $P_{SoD}^{a, a'}$. In both cases, with $\beta
      . \alpha \in P_{SoD}^{a, a'}$ we conclude that the sequence of user
      interactions in $\alpha'$, which is the same as in $\alpha$, is still
      possible after $c$, i.e.\ $\beta . c . \alpha' \in P_{SoD}^{a, a'}$,
      leading to a contradiction.
    \item 
      \begin{enumerate}
	\item Assume that all confidential user interaction events with $a$ or
	  $a'$ are disabled, i.e.\ $C \cap (E_a \cup E_{a'}) \cap E_U \subseteq
	  E^{disabled}$. Hence, $c \notin E_a$ and $c \notin E_{a'}$. With
	  $\alpha |_{E_U} = \alpha' |_{E_U}$ we have $\beta . c . \alpha' \in
	  P_{SoD}^{a, a'}$ iff $\beta . \alpha \in P_{SoD}^{a, a'}$, leading to
	  a contradiction.
	\item Assume that all visible user interaction events with $a$ or $a'$
	  are disabled, i.e.\ $V \cap (E_a \cup E_{a'}) \cap E_U \subseteq
	  E^{disabled}$. Hence, all user interaction events of activities $a$
	  and $a'$ in $\beta . c . \alpha'$ must be in $C$. However, as
	  $\alpha' |_{C} = \langle \rangle$, $\beta . c . \alpha' \in
	  P_{SoD}^{a, a'}$ iff $\beta . c \in P_{SoD}^{a, a'}$, leading to a
	  contradiction.
      \end{enumerate}
    \item Assume that $ES$ already enforces the separation of duty constraint,
      i.e.\ $P_{SoD}^{a, a'} \subseteq Tr$. Then $Tr \cap P_{SoD}^{a, a'} = Tr$
      and the conclusion trivially follows. \qedhere
  \end{enumerate}
\end{proof}

\end{techreport}

\end{document}